\documentclass{article}

\usepackage{amsfonts}
\usepackage{amssymb}
\usepackage{amsmath}
\usepackage{amsthm}

\usepackage{array}

\usepackage{tikz}

\usepackage{clrscode}   
\usepackage{authblk}    

\usepackage[utf8]{inputenc} 

\newtheorem{theorem}{Theorem}[section]
\newtheorem*{theorem*}{Theorem}

\newtheorem{lemma}[theorem]{Lemma}

\theoremstyle{remark}

\begin{document}

\title{An Efficient Algorithm for Permutation Iteration Using a Singly Linked List}
\date{January 2025}
\author[$\dagger$]{Thomas Baruchel}
\affil[$\dagger$]{\small Universit\'e de Pau et des Pays de l'Adour, France\authorcr\texttt{thomas.baruchel@univ-pau.fr}}
\maketitle



\abstract{\noindent We present a new content-agnostic algorithm for iterating over all permutations of a sequence. The algorithm leverages elementary~$O(1)$ operations on recursive lists. As a result, no new nodes are allocated during the computation. Instead, all elements are rearranged within the original nodes of the singly linked list throughout the process. While permutations are generated in an unusual order, the transitions between consecutive permutations remain smooth. A proof of concept written in the Lisp programming language is referenced and discussed.}
\\{\bf Keywords:} permutations, singly linked list, in-place algorithm

\maketitle

\section{Introduction}\label{sec1}

The most straightforward method for generating permutations is arguably the recursive approach. For each element $e$ in the set $A$, the permutations of the subset $A\setminus\{e\}$ are computed recursively. The element $e$ is then inserted at the beginning of each resulting sequence. Finally, the union of all computed parts is returned.

Basic implementations of this method generate all permutations in lexicographic order, but several variations can be employed. The well-known review by Robert Sedgewick enumerates these variations~\cite{sedgewick}. One of the most challenging aspects of such algorithms is selecting a new element $e$ at each step, as the recursive calls continuously reorder the list, intermingling used and unused elements.

The most commonly used data structure for the problem is an array. In this case, the new element $e_\textrm{new}$ is typically moved to the beginning of the sequence by swapping it with the previously selected element $e_\textrm{prev}$, which is a constant-time operation on arrays.

Since we focus on singly linked lists, our approach emphasizes extracting a carefully chosen element $e$ and inserting it at the beginning of the list. Furthermore, we constrain the extraction of a new element to specific locations within the list to avoid unnecessary iterations: either at the current working point or at the end of the list. As a result, the data structure is locally treated as a type of input-restricted deque in the proposed algorithm.

A surprisingly regular pattern enables elements to be popped from either the left (the current working point) or the right (the end of the list) while consistently selecting previously unused elements, despite the reordering of the list caused by recursive calls. That regular pattern allows to iterate over all permutations without inspecting the elements or their current positions. The algorithm's behavior is entirely determined by the recursive frame's dimensions, making it uniquely robust to any data type.

While similar algorithms have been used for years, particularly within the Lisp community, they often accumulate multiple in-place modifications during the recursion, resulting in consecutive permutations that may not be closely related. In contrast, the algorithm discussed below consistently transitions between consecutive permutations by picking a single element for inserting it elsewhere. The fact that no accumulation of operations occurs while descending and ascending the recursion stack allows us to implement a local deque with constant-time elementary operations, which is usually not trivial when working with singly linked lists.

Although illustrating the algorithm with a list as short as length 4 may not fully reveal the pattern it follows, it does provide an opportunity to observe the proximity between consecutive permutations:
{\footnotesize
\begin{center}
    \begin{tabular}{l@{\ $\rightarrow$\ }l@{\ $\rightarrow$\ }l@{\ $\rightarrow$\ }l@{\ $\rightarrow$\ }l@{\ $\rightarrow$\ }l}
$(0,1,2,3)$&$(0,1,3,2)$&$(0,3,1,2)$&$(0,3,2,1)$&$(0,2,3,1)$&$(0,2,1,3)$\\
$(2,0,1,3)$&$(2,0,3,1)$&$(2,3,0,1)$&$(2,3,1,0)$&$(2,1,3,0)$&$(2,1,0,3)$\\
$(3,2,1,0)$&$(3,2,0,1)$&$(3,0,2,1)$&$(3,0,1,2)$&$(3,1,0,2)$&$(3,1,2,0)$\\
$(1,3,2,0)$&$(1,3,0,2)$&$(1,0,3,2)$&$(1,0,2,3)$&$(1,2,0,3)$&$(1,2,3,0)$
    \end{tabular}\end{center} }

    The algorithm's skeleton is presented as pseudocode, followed by a proof of its correctness, primarily by demonstrating that each new element inserted at the working point was previously unused. We then analyze the distance between consecutive permutations, measured using a relevant metric, proving it to be near-optimal. Finally, we address implementation details and reference a working Lisp version of the algorithm. Lisp was chosen for its two key features: native support for singly linked lists and the mutability of its data structures.

\section{State of the art}

Given its intrinsic characteristics, our algorithm naturally falls into the category of combinatorial Gray codes, as it proceeds by performing a bounded number of elementary operations on a specific data structure. We recall the definition of a combinatorial Gray code as given by Torsten Mütze:
    \begin{quotation}
        A {\em combinatorial Gray code} is a listing of the combinatorial objects in the class of interest that contains each object exactly once such that any two consecutive objects in the list differ only by a ‘small change’~\cite{muetze}.
    \end{quotation}
As we work here with singly linked lists, the transition between two consecutive permutations arises from shifting a whole part of the list rather than from swapping two elements — this is still considered a combinatorial Gray code by Mütze. Several algorithms relying on `prefix shifts' are indeed mentioned in his review.

In order to take full advantage of the considered data structures and of recursion, `prefix shift' operations can easily be turned into `suffix shift' operations as well; picking the last element and inserting it elsewhere shares some similarities with these `prefix shift' algorithms, but our algorithm cannot be fully categorized as such, since it also involves a pick-and-insert operation in the middle of the list.

    Some algorithms are also devised to tackle Knuth's sigma-tau problem, and ours shares some similarities with them, since it uses two distinct kinds of pick-and-insert operations that closely resemble $\tau$ (swapping the two initial elements) and $\sigma$ (a circular left shift). In~\cite{liptak}, the algorithm relies on singly linked lists for achieving its goal in constant space. The two operations used here can be mapped directly onto their counterparts, except that they are local to the current recursive frame, rather than global as in Knuth's problem. And while we also care about working in-place in-order to leverage the classical properties of singly linked list, our use of recursion prevents any claim of $O(1)$ additional space.

Furthermore, a notable property of our algorithm is that it is fully content-agnostic: deciding which of the two operations is used at each step follows a mechanical pattern rather than a branching structure — all operations are applied blindly.

    By `blindly', we mean that the algorithm follows a fixed pattern when iterating over all permutations, without ever testing the elements themselves. This is made clear by the Lisp code in the appendix~\ref{lisp}: the only conditional branch in the recursive functions detects termination (on an empty list), while the general cases involve no branching at all (the odd/even cases being handled by mutual recursion instead).

    Thus, we do not intend to solve any pre-existing mathematical challenge (such as finding a Hamiltonian path on some transition graph, or solving the Sigma-Tau problem), but rather to provide a simple and robust new recursive algorithm. On the other hand, proving that its mechanical behavior consistently yields a new permutation at each step is, by itself, a notable property. In section~\ref{sec3bis}, we still prove that transitions `mostly' follow the edges of the bubble-sort graph\footnote{We refer here to $BS_n$, the bubble-sort graph of order $n$, as the graph joining permutations of $(1,2,\ldots,n)$ that are reachable from one another by swapping two adjacent elements.} — specifically, for 90\% of transitions.

\section{Description of the algorithm}\label{sec2}

    The following pseudocode relies on two operations performed on a list: extracting an element at a specified position (which simultaneously returns the element and modifies the list) and inserting an element at a specified position. Typographical differences emphasize the procedural nature of the insertion and the functional nature of the extraction, but this distinction holds little practical significance.

  For clarity, the pseudocode employs numeric positions to manipulate the list, highlighting how in-place modifications are performed. However, an implementation using a recursive data structure would instead traverse the list by iterating sequentially over its tail.

    Positions are assumed to be 0-indexed. The recursive procedure takes three arguments: a list, a current position (initialized to $0$ for the main call), and a callback function.
    {\footnotesize\begin{center}\begin{minipage}{.7\linewidth}
\begin{codebox}
    \Procname{$\proc{Permutations}(L,i,\proc{Func})$}
  \li $n \gets \func{length}(L)$
  \li   \If $i\geq n-1$
  \li   \Then
          $\proc{Func}(L)$
  \li   \Else
          $\proc{Permutations}(L, i+1, \proc{Func})$
  \li     $\proc{Insert}(L, i, \func{extract}(L, i+1))$
  \li     $\proc{Permutations}(L, i+1, \proc{Func})$
    \li \Repeat $(n-i-3)$ times
  \li   \If $n-i$ is even
  \li   \Then
      $\proc{Insert}(L, i, \func{extract}(L, n-1))$
  \li   \Else
            $\proc{Insert}(L, i, \func{extract}(L, i+1))$
        \End
  \li     $\proc{Permutations}(L, i+1, \proc{Func})$
        \End
  \li   \If $n-i>2$
  \li   \Then
            $\proc{Insert}(L, i, \func{extract}(L, i+1))$
  \li     $\proc{Permutations}(L, i+1, \proc{Func})$
\end{codebox}\end{minipage}\end{center}
    }

When entering this recursive function, we consider the elements of the list $L$ starting from index $i$, which form a sublist of length $n-i$. From this point onward, let $k$ denote the number of elements in this sublist. As expected, $k$ recursive calls are made in the general case, one for each new element inserted at position $i$:

\begin{itemize}
   \item the first recursive call (at line~4) will be referred to as the `no-op' call, as it allows the termination condition to be reached without modifying the current list, thereby yielding its initial (current) permutation;
   \item $k-1$ additional recursive calls are then made after inserting at position $i$ an element that was popped either from position $i+1$ or from the end of the list, according to the parity of $k$ and certain positional considerations.
\end{itemize}

    We refer to these recursive calls as `left' or `right', depending on the position of the previously popped element (more precisely, each of these terms denotes the two sequential operations: the list modification as well as the following recursive call). The general structure of the algorithm can now be illustrated by the following triangle:
{\footnotesize
\begin{center}
\begin{tabular}{l l l l l l l l l}
    $k=2$&no-op&left\\
    $k=3$&no-op&left&left\\
    $k=4$&no-op&left&right&left\\
    $k=5$&no-op&left&left&left&left\\
    $k=6$&no-op&left&right&right&right&left\\
    $k=7$&no-op&left&left&left&left&left&left\\
    $k=8$&no-op&left&right&right&right&right&right&left\\
    \ldots&\ldots&\ldots&\ldots&\ldots&\ldots&\ldots&\ldots&\ldots
\end{tabular}\end{center} }

    We illustrate that behaviour in the $k=6$ even case. Each line below gives an overview of the recursive call within the initial recursive frame, and only a few intermediate permutations (resulting from the recursive subcalls) are given. The rightmost permutation on each line matches the end of all these subcalls. The $k-1$ corresponding left/right operations are given as annotations:
    {\footnotesize
    \[
        \begin{array}{l@{\;\rightarrow\;}l@{\;\rightarrow\;}l@{\;\rightarrow\;}l@{\;\rightarrow\;}l@{\quad}l}
            012345& 012354& 012534& \cdots&
            021345&\textit{(followed with `left')}\\
            201345& 201354& 201534& \cdots&
210345&\textit{(followed with `right')}\\
            521034& 521043& 521403& \cdots&
512034&\textit{(followed with `right')}\\
            451203& 451230& 451320& \cdots&
415203&\textit{(followed with `right')}\\
            341520& 341502& 341052& \cdots&
314520&\textit{(followed with `left')}\\
            134520& 134502& 134052& \cdots&
143520
        \end{array}\]}

Although the pattern is simple and regular, it is not immediately evident that each popped element is inserted at position $i$ for the first time, as the intermediate recursive calls significantly reorder the elements in the sublist under consideration. The next section of this paper will analyze the pattern and prove the required property.

\section{Proof of correctness}\label{sec3}

The crucial idea here is to prove that each element is inserted at the working position exactly once before all permutations of the remaining elements are generated. That element can be seen as a temporary pivot in the working frame.

To understand how the sublist is modified during a recursive call, consider the last permutation produced while iterating over all permutations of $(0,1,2,\dots,k-1)$. Although the primary purpose of the recursive call is to generate multiple permutations, it can also be viewed externally as a direct mapping of $(0,1,2,\dots,k-1)$ to a specific resulting permutation:

{\centering\small
\begin{tabular}{l l}
    $k=0$&NIL\\
    $k=1$&$(0)$\\
    $k=2$&$(1,0)$\\
    $k=3$&$(1,0,2)$\\
    $k=4$&$(1,2,3,0)$\\
    $k=5$&$(1,0,2,3,4)$\\
    $k=6$&$(1,4,3,5,2,0)$\\
    $k=7$&$(1,0,2,3,4,5,6)$\\
    $k=8$&$(1,4,3,5,6,7,2,0)$\\
    $k=9$&$(1,0,2,3,4,5,6,7,8)$\\
    $k=10$&$(1,4,3,5,6,7,8,9,2,0)$\\
    $k=11$&$(1,0,2,3,4,5,6,7,8,9,10)$\\
    $k=12$&$(1,4,3,5,6,7,8,9,10,11,2,0)$\\
    \ldots&\ldots
\end{tabular}\par
}

\noindent In accordance with the pattern inherent in the algorithm, two distinct cases can be easily identified based on the parity of $k$, which will now be examined in two corresponding lemmas.

\begin{lemma}\label{lemma1}
Let $k$ be an odd number $(k>4)$, and assume that the last permutation of $k$ elements produced by the algorithm turns out to be the initial sequence with the first two elements swapped. When the algorithm is applied to the sequence $(0,1,2,\dots,k)$, each of the $k+1$ elements will be inserted exactly once at the beginning, and the final permutation generated by the algorithm will be $(1,4,3,5,6,7,8,\dots,k,2,0)$.
\end{lemma}
\begin{proof}
    As $k+1$ is even, the algorithm will make $k+1$ recursive calls following the pattern:
    \begin{center}
        no-op, left, right, right, right, \ldots, right, right, right, left
    \end{center}
    and according to the assumption, each recursive call will swap the second and the third elements, as the first element lies outside the scope of the recursive calls.

    Thus, the initial `no-op' call will leave the sequence of $k+1$ elements as $(0,2,1,3,4,5,6,\dots,k)$, and the first `left' call will transform it into $(2,1,0,3,4,5,6,\dots,k)$. The subsequent $k-2$ `right' calls will proceed as follows:
    \[
        \begin{array}{l}
            (k,1,2,0,3,4,5,\dots,k-1)\\
            (k-1,1,k,2,0,3,4,5,\dots,k-2)\\
            (k-2,1,k-1,k,2,0,3,4,5,\dots,k-3)\\
            \dots\\
            (3,1,4,5,\dots,k,2,0)
        \end{array}
    \]
    Finally the `left' call will produce the expected $(1,4,3,5,6,7,8,\dots,k,2,0)$. Clearly, each of the $k+1$ elements will have been inserted exactly once at the beginning of the sequence.\end{proof}

\begin{lemma}\label{lemma2}
    Let's $k$ be an even number $(k>4)$, and assume that the last permutation of $(0,1,2,\dots,k-1)$ produced by the algorithm turns out to be $(1,4,3,5,6,7,8,\dots,k-1,2,0)$. When the algorithm is applied to the sequence $(0,1,2,\dots,k)$, each of the $k+1$ elements will be inserted exactly once at the beginning, and the final permutation generated will be the initial sequence with the first two elements swapped.
\end{lemma}
\begin{proof}
    As $k+1$ is odd, the algorithm will make $k+1$ recursive calls following the pattern:
    \begin{center}
        no-op, left, left, left, left, left, left, left, left\ldots
    \end{center}

    By carefully relabeling the elements, we find that the initial `no-op' call leaves the sequence of $k+1$ elements as $(0,2,5,4,6,7,8,9,\dots,k,3,1)$.

    Next, $k$ `left' recursive calls are made. However, tracking the position of each element throughout the process is not as straightforward as in the proof of Lemma~\ref{lemma1}. An simpler approach can be achieved by leveraging graph theory.

    We begin by compressing the `left` recursive call into a single permutation $T$ of the $k+1$ elements:
    \begin{itemize}
        \item the two first elements are swapped initially;
        \item the permutation of $k$ elements from the assumption is then applied to the sublist starting at the second element.
    \end{itemize}
    
    Applying $T$ to $(0,1,2,\dots,k)$ results in $(1,2,5,4,6,7,8,\dots,k,3,0)$. We now represent the list as an array and describe $T$ as a directed graph based on the following concept: $k+1$ vertices correspond to each index in the array, and $k+1$ edges illustrate how each element moves from one position to another position when $T$ is applied. The resulting graph appears as follows:

    
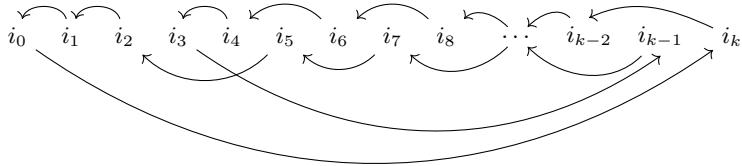
\begin{center}
        \begin{tikzpicture}[scale=0.95]
  \usetikzlibrary {graphs}
  \usetikzlibrary{arrows.meta}
  \node (i0) at (0,0) {$i_0$};
  \node (i1) at (.75,0) {$i_1$};
  \node (i2) at (1.5,0) {$i_2$};
  \node (i3) at (2.25,0) {$i_3$};
  \node (i4) at (3,0) {$i_4$};
  \node (i5) at (3.75,0) {$i_5$};
  \node (i6) at (4.5,0) {$i_6$};
  \node (i7) at (5.25,0) {$i_7$};
  \node (i8) at (6,0) {$i_8$};
  \node (i9) at (7,0) {$\dots$};
  \node (ia) at (8,0) {$i_{k-2}$};
  \node (ib) at (9,0) {$i_{k-1}$};
  \node (ic) at (10,0) {$i_{k}$};

        \draw[->] (i0) to[out=325, in=215] (ic);
        \draw[->] (i3) to[out=325, in=215] (ib.south);
        \draw[->] (ic) to[out=155, in=25] (ia.north);
        \draw[->] (ia) to[out=135, in=45] (i9);
        \draw[->] (ib) to[out=225, in=315] (i9);
        \draw[->] (i6) to[out=135, in=45] (i4);
        \draw[->] (i8) to[out=135, in=45] (i6);
        \draw[->] (i7) to[out=225, in=315] (i5);
        \draw[->] (i2) to[out=105, in=75] (i1);
        \draw[->] (i1) to[out=105, in=75] (i0);
        \draw[->] (i5) to[out=225, in=315] (i2);
        \draw[->] (i4) to[out=105, in=75] (i3);
        \draw[->] (i9) to[out=135, in=45] (i8);
        \draw[->] (i9) to[out=225, in=315] (i7);
    \end{tikzpicture}
    \end{center}
    \vspace{-20pt}
    
    Consider an element $e$ located at position $i_k$; the element $e$ will jump to $i_{k-2}$, $i_{k-4}$, \ldots, iterating over all even indices downward until it reaches $i_4$, and then $i_3$. Next, $e$ will move to $i_{k-1}$ and iterate over all odd indices downward, eventually reaching $i_5$. The final part of the path traverses $i_2$, $i_1$, and $i_0$. Observing the structure of the graph, we see that it forms a single strongly connected component, which ultimately reduces to a single directed cycle.

    It is now evident that, during the $k$ iterations, all elements will pass through $i_0$. One additional iteration would return the element initially at $i_0$ back to its starting position.

To compute the last permutation of the $k+1$ elements, instead of iterating $k$ times with $T$, it is more straightforward to perform a single {\em backward} iteration. Starting from $(0,2,5,4,6,7,8,9,\dots,k,3,1)$, as mentioned above, this leads to $(1, 0, 2, 3, 4, 5, \dots, k)$.

Thus, each element will be inserted at the beginning before a recursive call, and the list will return to its initial state, except for the first two elements being swapped.
\end{proof}

\begin{theorem*}
    The algorithm described in section~\ref{sec2} is correct.
\end{theorem*}
\begin{proof}
    The algorithm works correctly for $k$ in the range from $0$ to $5$. Additionally, we verify that the last computed permutation of $(e_0, e_1, e_2,e_3, e_4)$ is $(e_1, e_0, e_2,e_3, e_4)$. Using Lemmas~\ref{lemma1} and~\ref{lemma2}, we prove by induction that the algorithm selects each element of the sequence exactly once to insert it at the beginning before computing all permutations of the remaining elements.
\end{proof}

\section{Transition Distance in Permutations}\label{sec3bis}
Looking at the description of the algorithm provided in Section 2~\ref{sec2}, we can observe that no modifications occur before the first recursive call or after the last one. This leads to two resulting properties:
\begin{itemize}
    \item after each modification, the list is yielded without any further changes, as each recursive call waits until the halting condition is met before modifying the data structure for the first time;
    \item when a recursive call returns, the current state of the data structure remains unchanged since the last time the list was yielded.
\end{itemize}

The first implication is that a single modification occurs between two consecutive permutations: one element is selected from its position and inserted into another (in contrast to other similar algorithms). As a side note, in the specific case of an odd length, the same property also holds when {\em cycling} through the permutations of a list, as returning to the initial permutation from the last one merely requires swapping the first two elements.

A second implication is that calculating the distance between two consecutive permutations, according to a given metric, can be achieved through purely local considerations. Indeed, each transition is achieved by applying either a single `left' or `right' move, as defined earlier, with its length depending on the current size of the sublist.

The Kendall tau distance (which counts the number of swaps required for transitionning from a permutation to another one in the general case) is particularly relevant in this context, as it simplifies to the length of each move in our specific case—a property that is straightforward to formalize. The Kendall-tau distance is not only a fairly classical choice for studying our algorithm's behavior (allowing a direct comparison below with two other algorithms), but it also directly answers the question: what is the average length of the jumps in all the pick-and-insert operations involved here?

Let us denote by $D_k$ the cumulative distance in a complete traversal of all permutations of a list of length $k$. It is straightforward to observe that $D_k = k D_{k-1} + a_k$, where $k D_{k-1}$ represents the cumulative distance contributed by the $k$ recursive calls, and $a_k$ is an additional local term accounting for the cumulative distance of all modifications performed between the recursive calls.

Using the triangle that outlines the pattern of recursive calls in Section~\ref{sec2}, we derive the following recurrence relations:
\[
    \left\{
        \begin{array}{ll}
            D_1 = 0\\
            D_k = k D_{k-1}+1+(k-3)(k-1)+1 & \textrm{if $k$ is even}\\
            D_k = k D_{k-1}+(k-1)\times 1 & \textrm{if $k$ is odd}\\
        \end{array}
        \right.
\]

Solving these relations with the aid of a computer algebra system is feasible but yields cumbersome and inelegant formulas. However, employing these formulas to calculate the {\em average} distance between two consecutive permutations for large values of $k$ yields simple expressions, as several terms vanish in the limit:
\[
    \lim_{k\to\infty}\frac{D_k}{k!-1} = 7\cosh(1) - 4\sinh(1) - 5
    \quad
    \approx 1.100759669\dots
\]
which appears to be a remarkably low limit for a recursive algorithm of this kind. By comparison, Steinhaus–Johnson–Trotter algorithm has a transition distance of exactly~1 (which matches its best known feature), while the average distance between two consecutive permutations in lexicographic order approaches a limit of 2\footnote{While a formal proof of this result is beyond the scope of our study, an easy hint toward it can be given as follows. Consider a recursive algorithm iterating over all permutations in lexicographical order: starting from a list of length $n$, a frame of size $k$ is entered $n!/k!$ times; after each of its recursive calls except the last, the sublist of length $k-1$ is reverted at a cost of $(k-1)(k-2)/2$ swaps, plus one further swap to update the pivot. The total number of swaps is thus
\[
\sum_{k=1}^{n}\frac{n!}{k!}(k-1)\left(1+\frac{(k-1)(k-2)}{2}\right),
\]
which gives an average of 2 swaps per transition as $n$ grows.}.

In summary, each permutation is derived from the previous one by selecting an element from the current list and inserting it to the left, with an average movement of just over one position. In essence, about 90\% of the transitions reduce to swapping two consecutive elements.

This result is fairly easy to explain, as the conjunction of two factors: first, roughly half of the pattern followed by the algorithm consists of `left' changes (which amount to a single adjacent swap); second, most of the work in such a recursive algorithm is performed within small frames.

\section{Implementation technical details}\label{sec4}

        We now discuss an implementation of the algorithm in section~\ref{sec2} using singly linked lists. Such a data structure is typically not well-suited for extracting elements repeatedly from both the beginning and the end of a sublist. However, in this very specific case, a natural and elegant solution to this issue will emerge. We highlight that making this extraction an atomic operation is one of the key features of our algorithm.

We also require the data structure to be {\em mutable}, enabling a highly efficient implementation of our algorithm by moving and reusing nodes without creating new ones (eliminating the need for \texttt{cons} operations). This requirement led us to choose Lisp for the proof of concept; however, nothing in the code depends on Lisp-specific features, and it can easily be adapted to other programming languages. This second requirement is, of course, incompatible with adhering to a purely functional style, as described in Okasaki's renowned book~\cite{okasaki}, but it is essential for in-place operations.

Iterating over the consecutive sublists is very natural with such a data structure (by repeatedly taking the \texttt{cdr} of the list), and performing the `left' operation is straightforward, since it can be implemented either by swapping the values of the \texttt{car} and \texttt{cadr} or by directly manipulating the list structure.

We also note that extracting the last element of the sublist is never done before a recursive call, which leads to the following idea:
\begin{itemize}
    \item the descending phase of the recursive call is intended to compute new partial permutations as usual;
    \item the ascending phase of the recursive call is intended to pass back information needed by the next recursive call (returning a handle to the end of the list).
\end{itemize}

The figure below illustrates the current frame near the middle of the recursive call stack. An undefined element is added at the beginning of the list to ensure consistency across all calls. At this level, the sublist operates loosely as an input-restricted deque:

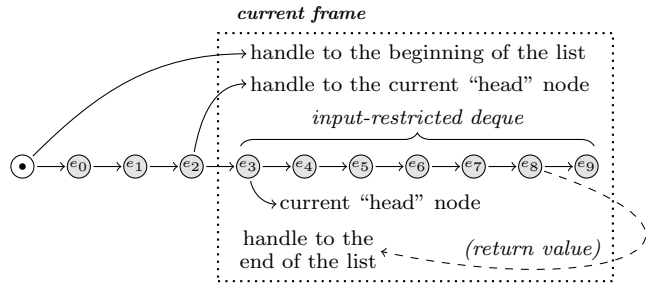
\hspace{8pt}
        \begin{tikzpicture}[main/.style = {draw, circle}]
  \usetikzlibrary{decorations.pathreplacing}
  \usetikzlibrary{arrows.meta}
    \node[main,fill={rgb:black,0;white,8}] (i0) at (0,0) {};
    \draw node[fill,circle,inner sep=0pt,minimum size=2.5pt] at (0,0) {};
    \node[main,fill={rgb:black,1;white,8}] (i1) at (0.75,0) {};
    \node (i1t) at (0.75,0) {\tiny $e_0$};
    \node[main,fill={rgb:black,1;white,8}] (i2) at (1.5,0) {};
    \node (i2t) at (1.5,0) {\tiny $e_1$};
    \node[main,fill={rgb:black,1;white,8}] (i3) at (2.25,0) {};
    \node (i3t) at (2.25,0) {\tiny $e_2$};
    \node[main,fill={rgb:black,1;white,8}] (i4) at (3,0) {};
    \node (i4t) at (3,0) {\tiny $e_3$};
    \node[main,fill={rgb:black,1;white,8}] (i5) at (3.75,0) {};
    \node (i5t) at (3.75,0) {\tiny $e_4$};
    \node[main,fill={rgb:black,1;white,8}] (i6) at (4.5,0) {};
    \node (i6t) at (4.5,0) {\tiny $e_5$};
    \node[main,fill={rgb:black,1;white,8}] (i7) at (5.25,0) {};
    \node (i7t) at (5.25,0) {\tiny $e_6$};
    \node[main,fill={rgb:black,1;white,8}] (i8) at (6,0) {};
    \node (i8t) at (6,0) {\tiny $e_7$};
    \node[main,fill={rgb:black,1;white,8}] (i9) at (6.75,0) {};
    \node (i9t) at (6.75,0) {\tiny $e_8$};
    \node[main,fill={rgb:black,1;white,8}] (i10) at (7.5,0) {};
    \node (i10t) at (7.5,0) {\tiny $e_9$};

        \draw[->, shorten >=.65pt,shorten <=1pt] (i0) to (i1);
        \draw[->, shorten >=.65pt,shorten <=1pt] (i1) to (i2);
        \draw[->, shorten >=.65pt,shorten <=1pt] (i2) to (i3);
        \draw[->, shorten >=.65pt,shorten <=1pt] (i3) to (i4);
        \draw[->, shorten >=.65pt,shorten <=1pt] (i4) to (i5);
        \draw[->, shorten >=.65pt,shorten <=1pt] (i5) to (i6);
        \draw[->, shorten >=.65pt,shorten <=1pt] (i6) to (i7);
        \draw[->, shorten >=.65pt,shorten <=1pt] (i7) to (i8);
        \draw[->, shorten >=.65pt,shorten <=1pt] (i8) to (i9);
        \draw[->, shorten >=.65pt,shorten <=1pt] (i9) to (i10);

        \draw[->, shorten >=.65pt,shorten <=1pt] (i0) to[out=45, in=180] (3,1.5);
        \draw[->, shorten >=.65pt,shorten <=1pt] (i3) to[out=80, in=180] (3,1.1);
            \draw[->, shorten >=.65pt,shorten <=1pt, dashed] (i9) .. controls+(3,-1) and +(3,-.65) .. (4.745,-1.125);

        \draw[->, shorten >=.65pt,shorten <=1pt] (i4) to[out=280, in=180] (3.4,-.5);

        \draw[draw=black, dotted, thick] (2.6,-1.525) rectangle ++(5.25,3.3);
    \node (t1) at (3.7,2.025) {\scriptsize\bfseries\em current frame};
    \node (t2) at (5.26,1.5) {\footnotesize handle to the beginning of the list};
    \node (t3) at (5.15,1.1) {\footnotesize handle to the current head node};
    \node (t6) at (4.65,-.47) {\footnotesize current head node};
    \node (t4) at (3.8,-.95) {\footnotesize handle to the};
    \node (t5) at (3.8,-1.25) {\footnotesize end of the list};
            \node (t5) at (6.775,-1.1) {\footnotesize\em (return value)};

            \draw [decorate,decoration={brace,amplitude=5pt,raise=2ex}] (2.9,0) -- (7.6,0) node[midway,yshift=1.95em]{\footnotesize\em input-restricted deque};
    \end{tikzpicture}

\vspace{-4pt}\noindent Full advantage is taken here of the observation made in Section~\ref{sec3bis}: only a single alteration is made to the list between two consecutive permutations. Consequently, the handle to the last element is used once, becoming immediately obsolete afterward, and no additional handle is needed until the next permutation is generated.

A snippet of code written in Lisp is given in appendix~\ref{lisp} and is available online at~\cite{baruchel}. The function performs a single \texttt{cons} operation (at line 27) to create a handle pointing to the actual beginning of the list, as the initial node may be repositioned during the process. The two auxiliary mutually recursive functions are subsequently invoked with either four or five arguments, depending on the parity of the current length:

\begin{itemize}
    \item the considered list \texttt{p};
    \item the callback function \texttt{func}, to be applied to each generated permutation;
    \item the length of the sublist decreased by~2 (in order to match the size of the loops without changing its parity);
    \item a handle to the previous node (before \texttt{p}) called \texttt{prev}, which always exists, since even the first useful element of the list has a preceding (\texttt{NIL}) element from the initial \texttt{cons} at line~27;
    \item a handle (also called \texttt{handle}) to the beginning of the entire list.
\end{itemize}

Since the modifications in the odd case are simpler than those in the even case, the \texttt{prev} parameter becomes redundant and was eliminated from that section of the code. Specifically, working on the left side of the list was achieved by swapping the contents of the two relevant nodes rather than physically relocating them.

The recursive functions return a handle to the penultimate node of the list, which allows extracting the last node to insert it before the current head node.

We chose $k=2$ for the halting case, as this specific case required slightly different handling compared to other even cases. Keeping the cases $k=0$ and $k=1$ outside the recursive process was straightforward.

The key point to note here is that, aside from the recursive calls, the function performs only elementary assignments using the \texttt{setf} or \texttt{rotatef} macros\footnote{The \texttt{setf} macro merely updates the content of a variable, while the \texttt{rotatef} macro does the same thing on several variables at once.}. Ensuring this was, in fact, the primary objective of the algorithm.

We also draw attention to two facts: the loops are mere `repeat' loops with no loop-control variable being updated, and the values in the list are never tested.

\section{Conclusion}\label{sec13}

In this paper, we introduced a novel algorithm for generating permutations using singly linked lists, emphasizing efficiency and elegance. Unlike conventional methods that rely on arrays or repeated node allocations, our approach leverages the mutability of linked lists to rearrange nodes dynamically without additional memory overhead. By utilizing a recursive structure and distinguishing between the descending and ascending phases of recursive calls, the algorithm efficiently handles both permutation generation and the propagation of structural information about the list.

The proof of correctness relies on a straightforward invariant: each element is inserted at the working position exactly once before all permutations of the remaining elements are generated. The transition distance analysis reveals that the average movement per element is just over one position, with approximately 90\% of transitions being adjacent swaps.

This work illustrates the potential for alternative data structures, such as singly linked lists, to solve traditionally array-dominated problems in a computationally efficient and conceptually elegant way. Future work may explore extending this approach to other combinatorial problems.

\section*{Declarations}

The author did not receive support from any organization for the submitted work.

The author received no financial support for the research, authorship, and/or publication of this article.

The author has no competing interests to declare that are relevant to the content of this article.

\appendix
        \section{Lisp implementation}\label{lisp}

        {\footnotesize
\begin{verbatim}
(defun permutations (li f)
  (labels
    ((aux-odd (p func n handle)
       (loop with s = (1- n)
             initially (aux-even (cdr p) func s p handle)
             repeat n
             do (progn (rotatef (car p) (cadr p))
                       (aux-even (cdr p) func s p handle))
             finally (return (progn
                       (rotatef (car p) (cadr p))
                       (aux-even (cdr p) func s p handle)))))
     (aux-even (p func n prev handle)
       (if (cddr p)
         (loop with s = (1- n)
               initially (progn (aux-odd (cdr p) func s handle)
                                (rotatef (car p) (cadr p)))
               repeat s
               do (let ((l (aux-odd (cddr prev) func s handle)))
                    (rotatef (cddr l) (cdr prev) (cdr l)))
               finally (return (progn
                    (setf p (cdr prev))
                    (aux-odd (cdr p) func s handle)
                    (rotatef (car p) (cadr p))
                    (aux-odd (cdr p) func s handle))))
           (progn (funcall func (cdr handle))
                  (rotatef (car p) (cadr p))
                  (funcall func (cdr handle)) p))))
  (let ((q (cons NIL li)) (n (length li)))
    (if (> n 1)
      (if (oddp n) (aux-odd li f (- n 2) q)
                   (aux-even li f (- n 2) q q))
      (funcall f li)) NIL)))
\end{verbatim}
}



\end{document}